\documentclass[11pt]{article}
\usepackage{algorithm}
\usepackage{algorithmic}

\usepackage{graphicx}
\usepackage{subfigure}
\usepackage{epsfig,amsthm,amsmath,color, amsfonts}
\usepackage{epsfig,color}

\usepackage{framed}
\newtheorem{theorem}{Theorem}[section]
\newtheorem{corollary}[theorem]{Corollary}

\newtheorem{lemma}[theorem]{Lemma}

\theoremstyle{definition}
\theoremstyle{definition}
\theoremstyle{observation}

\newcommand{\comment}[1]{}
\newcommand{\QED}{\mbox{}\hfill \rule{3pt}{8pt}\vspace{10pt}\par}

\def\polylog{\operatorname{polylog}}
\newcommand{\ignore}[1]{}
\newcommand{\eat}[1]{}

\usepackage{times}
\usepackage[small,compact]{titlesec}
\usepackage[small,it]{caption}

\newcommand{\squishlist}{
 \begin{list}{$\bullet$}
  { \setlength{\itemsep}{0pt}
     \setlength{\parsep}{3pt}
     \setlength{\topsep}{3pt}
     \setlength{\partopsep}{0pt}
     \setlength{\leftmargin}{1.5em}
     \setlength{\labelwidth}{1em}
     \setlength{\labelsep}{0.5em} } }
\newcommand{\squishend}{
  \end{list}  }

\def\anisur#1{}
\def\gopal#1{}
\def\atish#1{}

\begin{document}

\title{Near-Optimal Random Walk Sampling in Distributed Networks}

\begin{titlepage}
\author{Atish {Das Sarma} \thanks{Google Research, Google Inc., Mountain View, USA.
\hbox{E-mail}:~{\tt atish.dassarma@gmail.com}} \and  Anisur Rahaman Molla \thanks{Division of Mathematical
Sciences, Nanyang Technological University, Singapore 637371  \hbox{E-mail}:~{\tt anisurpm@gmail.com, gopalpandurangan@gmail.com}.  Supported in part by Nanyang Technological University grant M58110000.} \and Gopal Pandurangan \addtocounter{footnote}{-1} \footnotemark}

\date{}

\maketitle \thispagestyle{empty}

\vspace*{.4in}

\maketitle

\begin{abstract}
Performing random walks in networks is a fundamental primitive that has found  numerous applications in communication networks such as token management, load balancing, network topology discovery and construction,  search, and peer-to-peer membership management. While several such algorithms are ubiquitous, and use numerous random walk samples, the walks themselves have always been performed naively. 

In this paper, we focus on the problem of performing random walk sampling efficiently in a distributed network. Given bandwidth constraints, the goal is to minimize the number of rounds and messages required to obtain several random walk samples in a continuous online fashion. 
 We present the first round and message optimal distributed algorithms that present a significant improvement on all previous approaches. 
The theoretical analysis and comprehensive experimental evaluation of our algorithms show that they perform very well in different types of networks of differing topologies. 

In particular, our results show how several random walks can be performed continuously (when source nodes are provided only at runtime, i.e., online), such that each walk of length $\ell$ can be performed exactly in just $\tilde{O}(\sqrt{\ell D})$ rounds\footnote{Throughout this paper, $\tilde{O}$ hides polylogarithmic factors in the number of nodes in the network} (where $D$ is the diameter of the network), and $O(\ell)$ messages. This significantly improves upon both, the naive technique that requires $O(\ell)$ rounds and $O(\ell)$ messages, and the sophisticated algorithm of \cite{DasSarmaNPT10} that has the same round complexity as this paper but requires $\Omega(m\sqrt{\ell})$ messages (where $m$ is the number of edges in the network). Our theoretical results are corroborated through extensive experiments on various topological data sets. Our algorithms are fully decentralized, lightweight, and easily implementable, and can serve as building blocks in the design of topologically-aware networks.
\end{abstract}


\end{titlepage}

\vspace{-0.15in}
\section{Introduction}

Random walks play a central role in computer science, spanning a wide range of areas in both theory and practice, including distributed computing and communication networks. Algorithms in many different applications use random walks as an integral subroutine. Applications in communication networks include token management~\cite{IJ90, BBF04,CTW93}, load balancing~\cite{KR04}, small-world routing~\cite{K00}, search~\cite{ZS06,AHLP01,C05,GMS05,LCCLS02}, information propagation and gathering~\cite{BAS04,KKD01}, network topology construction~\cite{GMS05,LawS03,LKRG03}, checking expander~\cite{DT07}, constructing random spanning trees~\cite{Broder89, BIZ89, BFG+03}, monitoring overlays~\cite{MG07}, group communication in ad-hoc network~\cite{DSW06}, gathering and dissemination of information over a network \cite{AKL+79}, distributed construction of expander networks \cite{LawS03}, and peer-to-peer membership management~\cite{GKM03,ZSS05}. Random walks have also been used to provide uniform and efficient solutions to distributed control of dynamic networks \cite{BBSB04}. \cite{ZS06} describes a broad range of network applications that can benefit from random walks in dynamic and decentralized settings. For further references on applications of random walks to distributed computing and networks, see, e.g.~\cite{BBSB04,ZS06}.

A key purpose of random walks in network applications is to perform node sampling. Random walk-based sampling is simple, local, and robust. Random walks  also require little index or state maintenance which make them especially attractive to self-organizing dynamic networks such as Internet overlay and ad hoc wireless networks~\cite{BBSB04,ZS06}. In this paper we present  efficient distributed random walk sampling algorithms in networks that are significantly faster than the existing and naive approaches and at the same time achieve optimal message complexity. Our experimental results further show that our techniques perform very well in various network topologies. 

While the sampling requirements in different applications vary, whenever a true sample is required from a random walk of certain steps, all applications perform the walks naively --- by simply passing a token from one node to its neighbor: thus performing a random walk of length $\ell$ takes   time and messages that is linear with respect to $\ell$. Such an  algorithm may not scale well as the network size increases and hence it is better to investigate   algorithms with sublinear time and message complexity. Previous work in ~\cite{DasSarmaNPT10} shows how to (partially) overcome this hurdle through a quadratic improvement in time and perform random walks optimally, i.e. in $\tilde{O}(\sqrt{\ell D})$ rounds. However, their algorithm requires a large number of messages for every random walk, depending on the number of edges in the network. The algorithm presented here shows how to perform the walks with {\em optimal} message complexity, i.e. just $O(\ell)$ messages per walk amortized, without compromising at all on the worst case round complexity. Such algorithms can be useful building blocks in the design of {\em topologically (self-)aware} networks, i.e., networks that can  monitor and regulate themselves in a decentralized fashion.  (For
example,  efficiently computing the mixing time or the spectral gap,
allows  the network to monitor connectivity and expansion properties
of the network \cite{DasSarmaNPT10}.)
Further, the previous papers (\cite{DasSarmaNPT10, DNP09-podc}) only considered performing a single walk, or a few walks. Most applications, however, require several walks to be performed in a continuous manner. This continuous processing of walks is of specific importance in distributed networks and our results are applicable in this general framework.\\

\noindent {\bf Our Contributions}\\
{\bf 1.} We introduce the problem of continuous processing of random walks. The objective is for a network to support a continuous sequence of random walk requests from various source nodes and perform node sampling to minimize round and message complexity for each request. \\
{\bf 2.} We present the first algorithm that is efficient in both round complexity as well as message complexity. Our technique and analysis presents almost-tight bounds on the message and round complexity in a widely used network congestion model. \\
{\bf 3.} We perform comprehensive experimental evaluation on numerous topological networks and highlight the effectiveness and efficiency of our algorithm. The experimental results corroborate the theoretical contributions and show that our random walk sampling algorithm performs very well on various metrics for all parameter ranges.

\noindent {\bf Overview.} In the remainder of this section, we discuss the network model, related work, and the formal notation and problem formulation considered in this paper. We present our algorithms, and message and round complexity analyses in Section~\ref{sec:algos}. This rests on some concentration analysis of key random walk properties of our algorithm that we then prove in Section~\ref{sec:conc}. Finally, we present extensive experiments on a various topological networks in Section~\ref{sec:exp}.

\subsection{Distributed Network Model}
We model the communication network as an undirected, unweighted, connected $n$-node graph $G = (V, E)$. Every  node has limited initial knowledge. Specifically, assume that each node is associated with a distinct identity number  (e.g., its IP address). 
At the beginning of the computation, each node $v$ accepts as input its own identity number and the identity numbers of its neighbors in $G$. The node may also accept some additional inputs as specified by the problem at hand. The nodes are allowed to communicate through the edges of the graph $G$. We assume that the communication occurs in  synchronous  {\em rounds}. 
We will use only small-sized messages. In particular, in each round, each node $v$ is allowed to send a message of size $O(\log n)$ through each edge $e = (v, u)$ that is adjacent to $v$.  The message  will arrive to $u$ at the end of the current round. 
This is a  widely used  standard model to study distributed algorithms (e.g., see \cite{peleg, PK09}) and captures the bandwidth constraints inherent in real-world computer  networks . Our algorithms can be easily generalized if $B$ bits  are allowed (for any pre-specified parameter $B$) to be sent through each edge in a round. Typically, as assumed here, $B = O(\log n)$, which is number of bits needed to send a node id in a n-node network.

While this is a nice theoretical abstraction, it still does not motivate the most natural practical difficulties. A well established concern with this model is that for simple operations, the entire network may spawn a large number of parallel messages in order to minimize rounds. This can be very expensive from a practical standpoint. A critical component in the analysis of practical algorithms is the overall message complexity per execution of any algorithm. This becomes even more crucial from the standpoint of continuous processing of algorithms, perhaps even in parallel. Therefore, the goal is to design algorithms that have a low amortized message complexity and minimize the worst case round complexity, both simultaneously. Due to their conflicting nature, few algorithms perform well on both metrics. In this paper we present an algorithm that is near-optimal in terms of messages as well as rounds in parallel. 

\subsection{Related Work and Problem Statement}

\subsubsection*{Applications and Related Work}
Random walks have been used in a wide variety of applications in distributed networks as mentioned previously. We describe here some of the applications in more detail. 

Morales and Gupta~\cite{MG07} discuss about discovering a consistent and available monitoring overlay for a distributed system. For each node, one needs to select and discover a list of nodes that would monitor it. The monitoring set of nodes need to satisfy some structural properties such as consistency, verifiability, load balancing, and randomness, among others. This is where random walks
come in. Random walks are a natural way to discover a set of random nodes that are spread out (and hence scalable), that can in turn be used to monitor their local neighborhoods. Random walks have been used for this purpose in another paper by Ganesh et al.~\cite{GKM03} on peer-to-peer membership management for gossip-based protocols. Morales and Gupta~\cite{MoralesGupta09,MoralesGuptaTwo09} have several more papers in their line of work on AVMON system and similar systems that use several continuous node samples as a way to monitor distributed systems. 

Speeding up distributed algorithms using random walks has been considered for a long time. Besides our approach of speeding up the random walk itself, one popular approach is to reduce the {\it cover time}. Recently, Alon et. al.~\cite{AAKKLT} show that performing several random walks in parallel reduces the cover time in various types of graphs. They assert that the problem with performing random walks is often the latency. In these scenarios where many walks are performed, our results could help avoid too much latency and yield an additional speed-up factor.

A nice application of random walks is in  the design and analysis of expanders. We mention two results here. Law and Siu~\cite{LawS03} consider the problem of constructing expander graphs in a distributed fashion. One of the key subroutines in their algorithm is to perform several random walks from specified source nodes. 
Dolev and Tzachar~\cite{DT07} use  random walks to check if a given graph is an expander. The first algorithm given in \cite{DT07} is essentially to run a random walk of length $n\log{n}$ and mark every visited vertex. Later, it is checked if all vertices have been visited.
Broder~\cite{Broder89} and Wilson~\cite{Wilson96} gave algorithms to generate random spanning trees using random walks and Broder's algorithm was later applied to the network setting by Bar-Ilan and Zernik~\cite{BIZ89}. Recently Goyal et al.~\cite{GoyalRV09} show how to construct an expander/sparsifier using random spanning trees. A variety of such applications can greatly benefit using the random walk sampling techniques presented in this paper.

\subsubsection*{Notation and Problem Statement}
The basic problem we address is the following.
We are given an arbitrary undirected, unweighted, and connected $n$--node (vertex) network $G = (V,E)$ and a random walk request length $\ell$. The goal is to devise a distributed algorithm such that, the algorithm continuously accepts source node inputs $s$, and after each input, the algorithm performs a random walk of length $\ell$ and outputs the ID of a node $v$ which is randomly picked according to the probability that it is the destination of a random walk of length $\ell$ starting at $s$. Throughout this paper, we assume the standard random walk: in each step, an edge is taken from the current node $x$ with probability proportional to $1/deg(x)$ where $deg(x)$ is the degree of $x$. Our goal is to output a true  random sample from the $\ell$-walk distribution starting from $s$. Once the sample has been output, a new request is issued to the algorithm, and this proceeds in a continual manner. The objective is to minimize the round complexity as well as the message complexity for each of these requests. 

For clarity, observe that the following naive algorithm solves the above problem for a single request in $O(\ell)$ rounds and $O(\ell)$ messages: The walk of length $\ell$ is performed by sending a token for $\ell$ steps, picking a random neighbor with each step. Then, the destination node $v$ of this walk sends its ID back (along the same path) to the source for output. Our goal is to perform such sampling with significantly less number of rounds. At the other extreme is the series of work \cite{DNP09-podc,DasSarmaNPT10,NanongkaiDP11} where the round complexity of this single random walk request was heavily optimized by a cleverer algorithm. We mention details about this below, but in essence they proposed an approach to perform this walk in $\tilde{O}(\sqrt{\ell D})$ rounds but in exchange incurred a message complexity of $\Omega(m\sqrt{\ell})$. We would like to improve the messages significantly from here. In particular, we would like the best of both these two extreme worlds, and also support continuous requests in the process. 

The problem of performing just one walk was proposed in~\cite{DNP09-podc} under the name \textit{Computing One Random Walk where Source Outputs Destination (1-RW-SoD)} (for short, this problem will simply be called {\em Single Random Walk} in this paper), wherein the first sublinear time distributed algorithm was provided, requiring $\tilde{O}(\ell^{2/3}D^{1/3})$ rounds ($\tilde{O}$
hides $\polylog(n)$ factors); this improves over the naive $O(\ell)$ algorithm when the walk is long compared to the diameter (i.e., $\ell = \Omega(D \polylog n)$ where $D$ is the diameter of the network). This was the first result to break past the inherent sequential nature of random walks and beat the naive $\ell$ round approach, despite the fact that random walks have been used in
distributed networks for long and in a wide variety of applications. It was further conjectured in \cite{DNP09-podc} that the true number of rounds for this problem is $\tilde O(\sqrt{\ell D})$.

The high-level idea used in the $\tilde{O}(\ell^{2/3}D^{1/3})$-round algorithm in \cite{DNP09-podc} is to ``prepare'' a few short walks in the beginning (executed in parallel) and then carefully stitch
these walks together later as necessary. The same general approach was introduced in~\cite{AtishGP08} to find random walks in data streams with the main motivation of finding PageRank. However, the two models have very different constraints and motivations and hence the subsequent techniques used in \cite{DNP09-podc} and \cite{AtishGP08} are very different. The algorithms in \cite{DasSarmaNPT10} use the same general approach as \cite{DNP09-podc} but exploit certain key properties of random walks  to design even faster sublinear time algorithms; in particular, they show how a random walk can be performed in $\tilde{O}(\sqrt{\ell D})$ rounds. It was then shown in \cite{NanongkaiDP11} that these techniques are optimal in round complexity for performing a single random walk. 

None of these papers considered message complexity though, and did not consider the problem of continuously processing random walk requests. Our current paper is the first to ask about amortized message complexity under this continuous framework and our algorithms continue to hold worst case optimality in round complexity as well.


\section{Theoretical Analysis of Algorithms} \label{sec:algos}
\subsection{Algorithm descriptions}
We first describe the algorithm for single random walk in~\cite{DasSarmaNPT10} and then describe how to extend this idea for continuous random walks. The current algorithm is also randomized and we focus more on the message complexity. The high-level idea for single random walk is to perform many short random walks in parallel and later stitch them together~\cite{DNP09-podc, DasSarmaNPT10}. Then for multiple random walks we choose the source node randomly each time and perform single random walk using the same set of short length walks.     

Our main algorithm for performing continuous random walk each of length $\ell$ is described in {\sc Continuous-Random-Walk} (cf. Algorithm~\ref{alg:continuous-random-walk}). This algorithm uses other algorithms {\sc Pre-Processing} (cf. Algorithm~\ref{alg:pre-processing}) and {\sc Single-Random-Walk} (cf. Algorithm~\ref{alg:single-random-walk}). The {\sc Pre-Processing} function is called only one time at the beginning of {\sc Continuous-Random-Walk}, to perform $\eta d(v) \log n$ short walks of length $\lambda$ from each vertex $v$; once these pre-processed short walks are insufficient to answer a single random walk request, only then is the pre-processing table reconstructed and the algorithm resumes answering single random walk requests accessing the short length walks from the new table. At the end of {\sc Pre-Processing},  each vertex knows the destination IDs of the short walks
that it initiated.

\subsection{Previous Results - Rounds and Messages}
We first restate the main round complexity theory for {\sc Single-Random-Walk} and also state the message complexity of this algorithm. 
\begin{lemma} [Theorem $2.5$ in~\cite{DasSarmaNPT10}]
\label{thm:1-walk}
For any $\ell$, Algorithm {\sc Single-Random-Walk} (cf. Theorem $2.5$ in~\cite{DasSarmaNPT10}) solves the Single Random Walk Problem and, with probability at least
$1-\frac{2}{n}$, finishes in $\Theta\left(\lambda \eta \log{n} + \frac{\ell D}{\lambda}\right)$ rounds.
\end{lemma}

\newcommand{\mindegree}[0]{\delta}
\begin{algorithm}[]
\caption{\sc Pre-Processing($\eta$, $\lambda$)}
\label{alg:pre-processing}
\textbf{Input:} number of short walks of each node $v$ is $\eta deg(v)\log n$, and desired short walk lengths $\lambda$.\\
\textbf{Output:} set of short random walks of each nodes \\

\textbf{Each node $v$ performs $\eta_v=\eta \deg(v)\log n$
random walks of length $\lambda + r_i$ where $r_i$ (for each $1\leq
i\leq \eta$) is chosen independently at random in the range
$[0,\lambda-1]$.}
\begin{algorithmic}[1]
\STATE Let $r_{max} = \max_{1\leq i\leq \eta}{r_i}$, the random
numbers chosen independently for each of the $\eta_x$ walks.

\STATE Each node $x$ constructs $\eta_x$ messages containing its ID
and in addition, the $i$-th message contains the desired walk length
of $\lambda + r_i$.

\FOR{$i=1$ to $\lambda + r_{max}$}

\STATE Each node $v$ does the
following: Consider each message $M$ held by $v$ and received in the
$(i-1)$-th iteration (having current counter $i-1$). If the message
$M$'s desired walk length is at most $i$, then $v$ stored the ID of
the source ($v$ is the desired destination). Else, $v$ picks a
neighbor $u$ uniformly at random and forward $M$ to $u$ after
incrementing its counter.
%

\ENDFOR

\STATE Send the destination IDs back to the respective sources (this can be done
by sending the destination IDs along the "reverse" path).

\end{algorithmic}
\end{algorithm}

\begin{algorithm}[]
\caption{\sc Single-Random-Walk($s$, $\ell$)}
\label{alg:single-random-walk}
\textbf{Input:} Starting node $s$, and desired walk length $\ell$.\\
\textbf{Output:} Destination node of the walk outputs the ID of $s$.\\

\textbf{Stitch $\Theta(\ell/\lambda)$ walks, each of length in $[\lambda,2\lambda-1]$ }
\begin{algorithmic}[1]
\STATE The source node $s$ creates a message called ``token'' which
contains the ID of $s$

\STATE The algorithm generates a set of {\em connectors}, denoted by
$C$, as follows. (Connectors are the endpoints of the short walks, i.e., the points where we stitch.)

\STATE Initialize $C = \{s\}$

\WHILE {Length of walk completed is at most $\ell-2\lambda$}

  \STATE Let $v$ be the node that is currently holding the token.

  \STATE $v$ uniformly chooses one of its short length sample and let $v'$ be the
  sampled value if any exists (which is a destination of an unused random walk of length between $\lambda$ and $2\lambda-1$). 

  \IF{$v' =$ {\sc null} (all walks from $v$ have already been used up)}


  
  \STATE Algorithm terminates failing this walk.

  \ENDIF

  \STATE $v$ sends the token to $v'$

  \STATE $C = C \cup \{v\}$

\ENDWHILE

\STATE Walk naively until $\ell$ steps are completed (this is at
most another $2\lambda$ steps)

\STATE A node holding the token outputs the ID of $s$

\end{algorithmic}

\end{algorithm}

\begin{algorithm}[H]
\caption{\sc Continuous-Random-Walk($\ell$)} \label{alg:continuous-random-walk}
\textbf{Input:} $\ell$.\\
\textbf{Output:} Continuous $\ell$ length random walk samples from sources nodes presented adversarially or randomly.\\ 
\textbf{Source nodes $S$:} The source node of each walk of length $\ell$ can be presented adversarially or randomly accordingly to some distribution. Let this continuous sequence of source nodes be denoted by ordered set $S$.
\begin{algorithmic}[1]

\STATE Call {\sc Pre-Processing($\eta=1$, $\lambda = 24\sqrt{\ell D} (\log{n})^3$)}

\WHILE{Indefinitely}

\WHILE{Algorithm does not fail (algorithm gets stuck due to insufficient short walks)}

\STATE Select the next source node $s$ from the ordered set $S$. 

\STATE call {\sc Single-Random-Walk($s$, $\ell$)}.

\STATE If Single-Random-Walk returns with fail, exit loop

\ENDWHILE

\STATE Call {\sc Pre-Processing($\eta=1$, $\lambda = 24\sqrt{\ell D} (\log{n})^3$)} again and use this table.

\STATE Rerun request for $s$ and then continue subsequent walks based on random samples. 

\ENDWHILE

\end{algorithmic}

\end{algorithm}

\begin{lemma}\label{thm: message-1-walk} The message complexity of {\sc Single-Random-Walk} is
$O\left(\eta \lambda m \log n + \frac{ \ell D}{\lambda} \right)$ where $m$ is number of edges and $D$ is the diameter of the network.  
\end{lemma}
\begin{proof}
For computing $\eta deg(v) \log n$ short walks of length $\lambda$ it uses $\Theta(\lambda \eta deg(v) \log n)$ messages. Since for a single short walk of length $\lambda$ it sends $\lambda$ messages and hence for $n$ nodes it requires $\Theta(\lambda \eta \log n \sum_v{deg(v)}) = \Theta(\lambda \eta m \log n)$ messages. 
For stitching one short walk with another we need to contact the destination ID. This can be done quickly by using a BFS tree.  Note that  the BFS tree needs to be constructed only once \footnote{If we assume that nodes
have access to shortest path routing table, then BFS tree is not needed.} ($\Theta(m)$ messages) and each stitch  uses $O(D)$ messages. Combining these, the lemma follows. 
\end{proof}

In networks such as P2P or overlay networks, if we assume that a node can access quickly (in constant time)
 another node whose ID (IP address) is known, then one can improve the time and message complexity of stitching, saving a $\Theta(D)$ factor.

We now analyze the round and message complexity for {\sc Continuous-Random-Walk} algorithm in the next two subsections. For simplified analysis, we use $\kappa$ to denote the fraction of short walks of the pre-processing table that get used, before the algorithm fails and needs to rerun the pre-processing stage. The next two subsections assume a value of $\kappa$ and prove bounds using it. In the following section, we actually present bounds on $\kappa$ itself to arrive at the main theorem of this paper. To recall other notation, $\eta_v = \eta deg(v) \log n$ is the number of short length walks pre-processed for each node $v$, $\lambda$ is the length of these short walks, $n$ is the number of nodes, $m$ is the number of edges, and $D$ is the diameter of the network.

\subsection{Round Complexity}
\begin{lemma}\label{thm:round-multi-walk}
For any $\ell$, Algorithm {\sc Continuous-Random-Walk} (cf. Algorithm~\ref{alg:continuous-random-walk}) serves continuous random walk requests such that, with probability at least
$1-\frac{2}{n}$, the total number of rounds used until {\sc Pre-Processing} needs to be invoked for a second time is $O\left(\lambda \eta \log{n} + \kappa m \eta D \log n \right)$, where $\kappa$ is the fraction of used short length walks from the preprocessing table.
\end{lemma}
\begin{proof}
The proof is same as Theorem 2.5 in~\cite{DasSarmaNPT10} for Single Random Walk; the only difference is we are doing continuous walks of same length $\ell$.  Therefore for Continuous Walks, if $\kappa$ is the fraction of used short length walks form the preprocessing table, then a total $O(\kappa m \eta \log n)$ short walks are used. Hence  we need to stitch  $O(\kappa m \eta \log n)$ times  and therefore by Lemma 2.3 in~\cite{DasSarmaNPT10}, contributes $O(\kappa m \eta \log n D)$ rounds. Hence total $O\left(\lambda \eta \log{n} + \kappa m \eta D \log n \right)$ rounds. 
\end{proof}

\begin{corollary}\label{thm:avg-round}
The average number of rounds per random walk of length $\ell$ of {\sc Continuous-Random-Walk} (cf. Algorithm~\ref{alg:continuous-random-walk})  is $ O\left( \frac{\ell}{ \kappa} (\frac{\log{n}}{m} + \frac{\kappa D}{\lambda}) \right)$ with high probability.
\end{corollary}
\begin{proof}
The total number of random walks of length $\ell$ that have been completed successfully by {\sc Continuous-Random-Walk} is $\Theta\left(\frac{\kappa m \eta \lambda \log n}{\ell}\right)$, as total $O(\kappa m \eta \log n)$ short walks each of length $\lambda$ have been used. Hence the bound on the average number of rounds per walk follows.
\end{proof}

\subsection{Message Complexity}
\begin{lemma}\label{thm:message-complexity1} The message complexity of {\sc Continuous-Random-Walk}, until {\sc Pre-Processing} needs to be invoked for a second time, is
$O\left(\eta \lambda m \log n + \kappa m \eta D\log n \right)$ where $\kappa$ is fraction of used short length walks from the preprocessing table. 
\end{lemma}
\begin{proof}
The message complexity of the stage of {\sc Pre-Processing} is as before. Further, for each subsequent $\ell$ length walk request, an additional $O(D \ell/\lambda)$ messages are used. Also, as before we know that the total number of random walks of length $\ell$ that have been completed successfully by {\sc Continuous-Random-Walk} is $\Theta\left(\frac{\kappa m \eta \lambda \log n}{\ell}\right)$, as total $O(\kappa m \eta \log n)$ short walks each of length $\lambda$ have been used. Therefore the contribution from this towards the total message complexity is $O(D\ell/\lambda * \frac{\kappa m \eta \lambda\log n}{\ell})$ which reduces to $O(mD\eta\kappa\log n)$. Combining these, the lemma follows.
\end{proof}

\begin{corollary}\label{thm:avg-message-complexity} The average number of messages per random walk of length $\ell$ of {\sc Continuous-Random-Walk} is
$ O\left( \frac{\ell}{ \kappa} (1 + \frac{\kappa D}{\lambda}) \right)$. 
\end{corollary}
\begin{proof}
From the above Lemma~\ref{thm:message-complexity1} we know that the total number of messages used for computing all walks of  {\sc Continuous-Random-Walk} is $O\left(\eta \lambda m \log n + \kappa m \eta D \log n \right)$. Now the total number of walks of length $\ell$ is $O\left(\frac{\kappa m \eta \lambda \log n}{\ell}\right)$, as total $O(\kappa m \eta)$ short walks each of length $\lambda$. Hence we get the average number of messages per walk by dividing by this.
\end{proof}

Combining the above two corollaries, we get the following. 

\begin{lemma}\label{thm:combined-avg-complexity}
The average number of rounds and messages per random walk of length $\ell$ of {\sc Continuous-Random-Walk} (cf. Algorithm~\ref{alg:continuous-random-walk}) are $ O\left( \frac{\ell}{\kappa} (\frac{\log{n}}{m} + \frac{\kappa D}{\lambda}) \right)$ and $ O\left( \frac{\ell}{ \kappa} (1 + \frac{\kappa D}{\lambda}) \right)$ respectively. 
\end{lemma}

\begin{corollary}\label{cor:avg-complexity}
For our choice of $\lambda = \tilde{\Theta}(\sqrt{\ell D})$, the average rounds and messages per random walk becomes $\tilde{O}\left( \frac{\ell}{ \kappa m}  + \sqrt{\ell D} + D \right)$ and $O\left( \frac{\ell}{ \kappa}  + D \right)$ respectively.
\end{corollary}
\begin{proof}
If we put $\lambda = \tilde{\Theta}(\sqrt{\ell D})$ in Lemma~\ref{thm:combined-avg-complexity} then the average round and message becomes  $ \tilde{O}\left( \frac{\ell}{ \kappa m}  + \sqrt{\ell D} \right)$ and $ O\left( \frac{\ell}{ \kappa}  + \sqrt{\ell D} \right)$ respectively. Now $\sqrt{\ell D} \leq \ell + D$.  So the corollary follows. 
\end{proof}

Note that, in the above corollary, $\kappa < 1$ can be small, so that the bounds can become large. We show in the next section that $\kappa$ is a constant and hence our bounds are 
almost optimal.

\section{Concentration Bounds on $\kappa$} \label{sec:conc}

The goal of this section is to present a lower bound on $\kappa$, the fraction of rows of the pre-processing table (or the fraction of all short walks) that get used before the algorithm fails to perform a random walk request, and needs to rerun the pre-processing stage.  All the analysis in this section assumes that sources $S$ in {\sc Continuous-Random-Walk} are sampled according to the degree distribution. While the algorithm {\sc Continuous-Random-Walk} remains meaningful otherwise also, our proofs crucially rely on this random sampling of sources. Obtaining similar theorems for more general sequence of sources in $S$ remains an open question. 

We now present the central theorem that lower bounds $\kappa$.

\begin{theorem} \label{thm:kappabound}
Given any graph $G$, if {\sc Continuous-Random-Walk} is invoked on $\ell = O(m)$ and the source nodes $S$ are chosen randomly proportional to the node degrees, then the algorithm uses up at least $\kappa = \Omega(1)$ fraction of all short walks in {\sc Pre-Processing} table, before a request fails and a second call needs to be made to {\sc Pre-Processing}. 
\end{theorem}
\begin{proof}
Assume for now that we do $d(v)$ short walks for each vertex $v$. The total number of walks of length $\ell$ is $ T = \frac{2m \lambda}{\ell}$ if all the short walks are used. 
Let $K = \alpha T$, where $\alpha$ is a constant in $[0, 1]$. 
Note that if we manage to perform  $K$ walks of length $\ell$, then we have utilized a constant fraction
of the short walks. For one $\ell$-length walk, in expectation a vertex $v$ can be a connector at most $\frac{d(v) \ell}{2m \lambda}$ times (by linearity of expectation). (Connectors are the endpoints of the short walks, i.e., the points where we stitch. Note that only when
a vertex is visited as a connector we end up using a short walk initiated from that vertex.) Then for $K$ walks, each of length $\ell$, the expected number of times that $v$ is visited as a connector vertex is $K \frac{d(v) \ell}{2m \lambda} = \alpha d(v)$. Let $N$ denote the number of times the vertex $v$ is visited as a connector in $K$ walks. By above, $ E[N] = \alpha d(v) $. 
By Markov's inequality, $\Pr\left(N \geq d(v) \right) \leq \frac{\alpha d(v)}{d(v)} = \alpha$.
 Now consider the above experiment (for a fixed vertex $v$) repeated $c\log n$ independent times for some constant $c$, that is suitably large. (In other words, assume that we do  $c d(v) \log n$ short walks --- total over
 all experiments --- from each node $v$.)
We say that an experiment is  "success'' if $N < d(v)$.  If we have success,
then that means that we have done $K$ walks of length $\ell$ (and hence utilized a constant fraction of the $d(v)$ short walks) for that experiment before a request fails. By above,
the probability of success is at least some constant $\alpha' = 1 - \alpha$. 
 Let $X_1^v, X_2^v, \ldots, X_{c\log n}^v$ be the 0-1 indicator random variables such that 
$X_i^v = 1$ (if success occurs in $i$-th time)  and zero otherwise.
Let $X^v = \sum_{1=1}^{c \log n} X_i^v $. Then $E[X^v]  = \alpha' c \log n$. Since the variables are independent, by Chernoff's bound
$ \Pr(| X^v - E[X^v] | \geq c' \log n) \leq  e^{-\frac{2c'^2\log^2 n}{c \log n}} \leq \frac{1}{n^2}$,
for a suitable constant $c \leq (c')^2$.
Therefore,  $ \Pr(| X^v - E[X^v] | \geq c' \log n) \leq \frac{1}{n^2}$.  Thus, at least a constant
fraction of the experiments succeed with probability $1 - 1/n^2$.    By union bound \cite{MU-book-05},
the total number of visits to every vertex $v$ as connector in all  ($c \log n$  times $K$) walks is at most $O(\alpha d(v) c \log n)$ with probability at least $1 - 1/n$.  This implies that the total number of short walks utilized is
a constant fraction of the best possible, before a request fails.
\end{proof}

We now present the main theorem of this paper, which follows from the above bound on $\kappa$ stated in Theorem~\ref{thm:kappabound}, and the message and round complexity bounds in terms of $\kappa$ stated in Corollary~\ref{cor:avg-complexity}. Notice that this presents optimal round and message complexities simultaneously for every walk (since independently also $\Omega(\ell + D)$ is a clear lower bound on the number of messages for a single $\ell$-length random walk, and $\Omega(\sqrt{\ell D} + D)$ is a nontrivial lower bound on the number of rounds for a single $\ell$-length random walk as shown in~\cite{NanongkaiDP11}).

\begin{theorem}
Algorithm {\sc Continuous-Random-Walks} satisfies walk requests continuously and indefinitely such that the amortized message complexity per walk is $O(\ell + D)$, and, with high probability, every single walk request completes in $\tilde{O}(\sqrt{\ell D} + D)$ rounds.
\end{theorem}

\subsection{Extensions to different walk lengths}

While our main algorithm of {\sc Continuous-Random-Walk} and the associated theorems are stated for a fixed $\ell$, they can call be generalized to handle different walk lengths. We omit the rigorous details for brevity and present a brief explanation of the generalization here. The theorems and experiments go through verbatim for this case as well.

Suppose that {\sc Continuous-Random-Walk} is designed to not only support new source node requests each time but also new length requests for the random walks. 
One can of course store multiple {\sc Pre-Processing} tables, one for each associated $\ell_i$, in the entire allowed range for $\ell$. This way, when a request is presented, the appropriate {\sc Pre-Processing} table is accessed and the corresponding short walks queried. Then, whenever {\sc Continuous-Random-Walk} fails on a specific single random walk request, 
only this {\sc Pre-Processing} is rerun, and answering the random walk requests resume. 

While this does solve the problem and guarantees the identical throughput and efficiency, a practical concern is that performing and storing so many short walks, corresponding to multiple different lengths, can be expensive. There is a simple way to counter this, by storing short walks in a doubling fashion. In particular, if each $\ell_i$ was in the range $[1, n]$, instead of storing short walks corresponding to each of $\ell_i = 1, 2, 3, \ldots, n-1, n$, we perform short walks only corresponding to $\ell_i = 1, 2, 4, \ldots, n/2, n$. 
This exponentially reduces the number of short walks at each node, or the number of pre-processing tables, from $n$ to $\log n$. Now, whenever a walk request for $\ell_i$ is received, it can be answered by just performing a longer walk, of length $\tilde{\ell}_i$ such that $\ell_i\leq \tilde{\ell}_i\leq 2\ell_i$.

\section{Experiments} \label{sec:exp}
We have used the following five important graph generative models for experiments. Several of these have been used in other papers as well for random walk experiments, see for e.g. \cite{GkantsidisMS04}. These graphs together cover a nice spectrum of fast mixing to slow mixing, uniform degrees to very skewed degrees, small diameter to large diameter, etc. thereby testing the algorithm in all the extreme cases as well as nice cases.  

\begin{itemize}
\item Regular Expander: We worked on the most commonly studied random graph model of $G(n,p)$. Here, each of the $n(n-1)/2$ edges occurs independently and randomly with probability $p$. We choose $p$ as $\log n/n$ so that the expected number of edges is roughly $(n\log n)/2$. Further, the expected degree of every vertex is $\log n$. This, with high probability, results in a graph with good expansion and it is regular in expectation. 
\item Two-tier topologies with clustering: First we construct four isolated roughly regular expanders, as mentioned above in $G(n,p)$, of the same size - think of these as independent clusters. Then from each cluster we pick a small number of nodes (roughly one-fourth the size of the cluster and connect them using another $G(n,p)$ - think of this as a tier-two cluster. Again we use the same value of $p$ as above.   
\item Power-law graphs: In distributed settings, many important networks are known to have power-laws. We use the well known preferential attachment growth model to construct random power-law graphs. The essential process proceeds by starting with a small clique (of same 5 nodes), and then adding vertices sequentially. Each subsequent vertex added connects with an edge to each of the previous edges with probability depending on their degrees, and independently. Specifically, the new vertex connects with a previous vertex $v$ with probability proportional to $deg(v)^{\alpha}$ where the exponent $\alpha$ is a parameter. 
\item Random Geometric Graph: A random geometric graph is a random undirected graph drawn on a bounded region $[0,1)\times [0,1)$. It is generated by placing $n$ vertices uniformly at random and independently on the region (i.e. both the $x$ and $y$ coordinates are picked uniformly and independently). Then edges are constructed deterministically - two vertices $u$ and $v$ are connected by an edge if and only if the distance between them is at most a parameter threshold $r$. We choose $r$ as $\sqrt{\frac{\log n}{n}}$ so that the degree of each vertices is $O(\log n)$ w.h.p. 
\item Grid Graph: Consider a square grid graph ($\sqrt{n}\times \sqrt{n}$) which is a Cartesian product of two path graphs with $\sqrt{n}$ vertices on each. Since a path graph is a median graph, the square grid graph is also a median graph. All grid graphs are bipartite (since they have no odd length cycles).
\end{itemize}

We compute and maintain a preprocessing table containing $\eta_v = \eta deg(v)\log n$ short walks of length $\lambda$ from each vertex $v$. We then check how many walks of length $\ell$ can be done using this table before we hit a node all of whose short walks have been exhausted. The source nodes for each of the $\ell$ length random walk requests are sampled randomly according to the degree distribution. 

We perform experiments on each of the aforementioned synthetically generated graphs, and also by varying different parameters. In particular, we conduct separate experiments for each of (a) varying the length of the walk ($\ell$) as a function of $n$, (b) varying the number of nodes($n$), (c) varying the length of the short walk ($\lambda$) stored by the preprocessing table, and (d) varying the number of short walks stored from each node as a function of the parameter $\eta$. For each of these, we use certain default values when a specific parameter is being varied for a plot, while others are held constant. The default values we use are $n = 10,000$, $\ell = n$, $\eta = 1$, and $\lambda = \log n$. 

Since we are interested in how many random walks of length $\ell$ can be done in a continuous manner with small round and message complexity, this translates to analyzing the utilization for one specific pre-processing table before {\sc Continuous-Random-Walk} gets stuck and needs to invoke another call to {\sc Pre-Processing}. In particular, to analyze the round complexity, we conduct a set of experiments to evaluate $\kappa$, the number of rows of the {\sc Pre-Processing} table used before the algorithm fails ($\kappa$ plotted on the $y$-axis). As mentioned in the previous section, this gives a bound of $\ell/\kappa$ on the round complexity. In particular, if $\kappa$ is a constant, and large enough, this shows excellent utilization and an asymptotically optimal round complexity. Similarly, for message complexity, we explicitly conduct a second set of plots that calculates the message complexity on the $y$-axis based on $\kappa$ and $D$, for easier visualization.  

We plot graphs by varying each of the parameters $\ell, n, \lambda, \eta$. Each figure contains fives lines, one for each of the above network models: For each of these plot values, we perform ten different runs and then present the average value. 

\subsection{Short walk utilization factor $\kappa$}

\begin{figure}[htbp]
  \centering
  \includegraphics[width=0.6\linewidth]{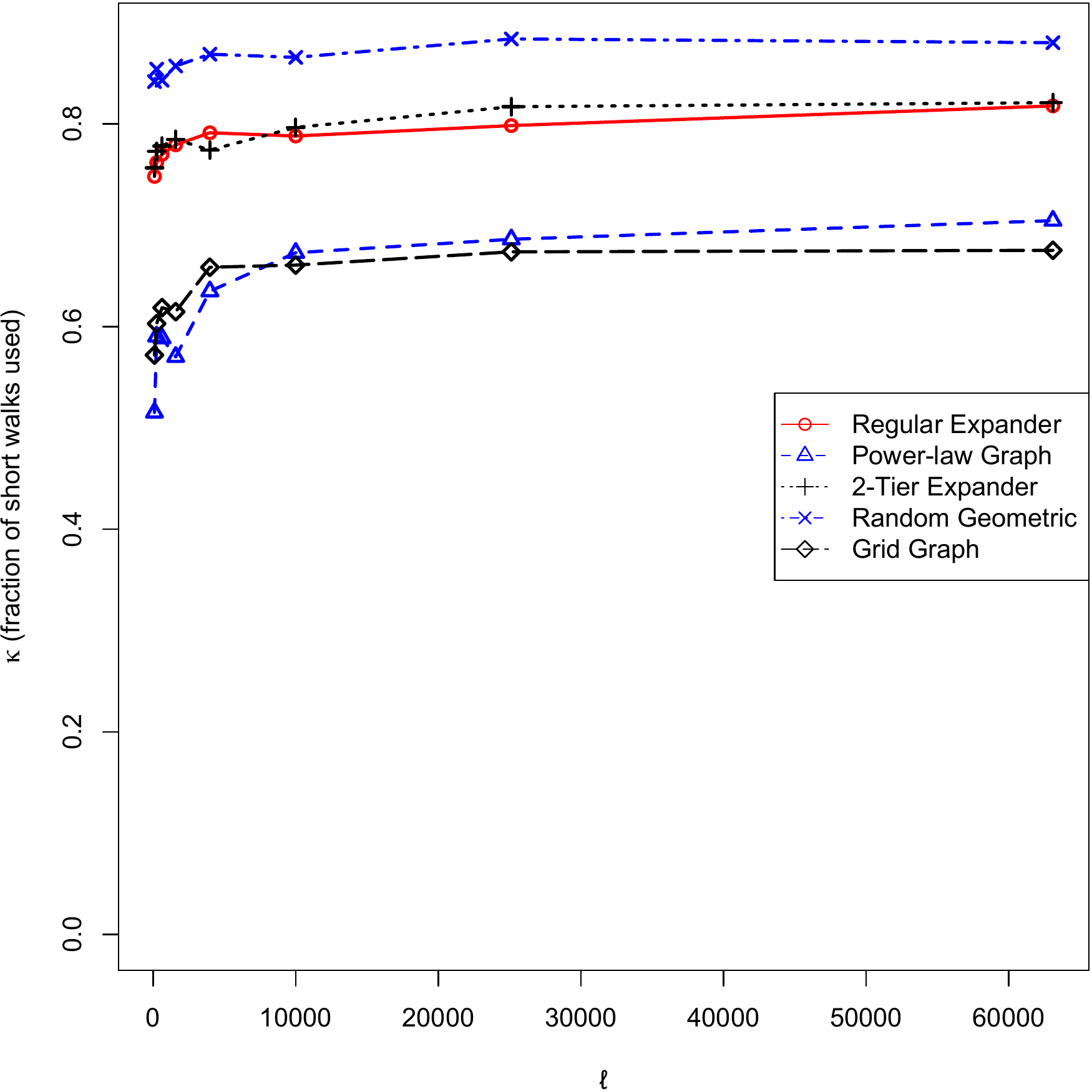}\\
  \caption{varying length of the walk $\ell$. $n=10K, \eta = 1, \lambda = \sqrt{\ell}$}
  \label{fig:plot1c}
\end{figure}

\noindent{\bf Varying $\ell$} [Figure~\ref{fig:plot1c}]: Here $n$ is fixed at 10,000 and $\ell$ is varying as $n^{0.5}, n^{0.6}, ... ,n^{1.2}$; $\lambda$ is $\sqrt{\ell}$ and $\eta$ is $\log n$.  In this case we see that at least 50\% of the pre-processed short walk rows are used up. This utilization is even better for some of the graph topologies such as $G(n,p)$ and two-tier clustering graph and reaches around 80\%. Therefore, for the entire range of $\ell$ being small to very large, our algorithm performs extremely well: In particular, $\kappa$ is a large constant and therefore the round complexity and message complexity are close to optimal - i.e. within a constant factor of the best possible. 
%
%

\begin{figure}[htbp]
  \centering
  \includegraphics[width=0.6\linewidth]{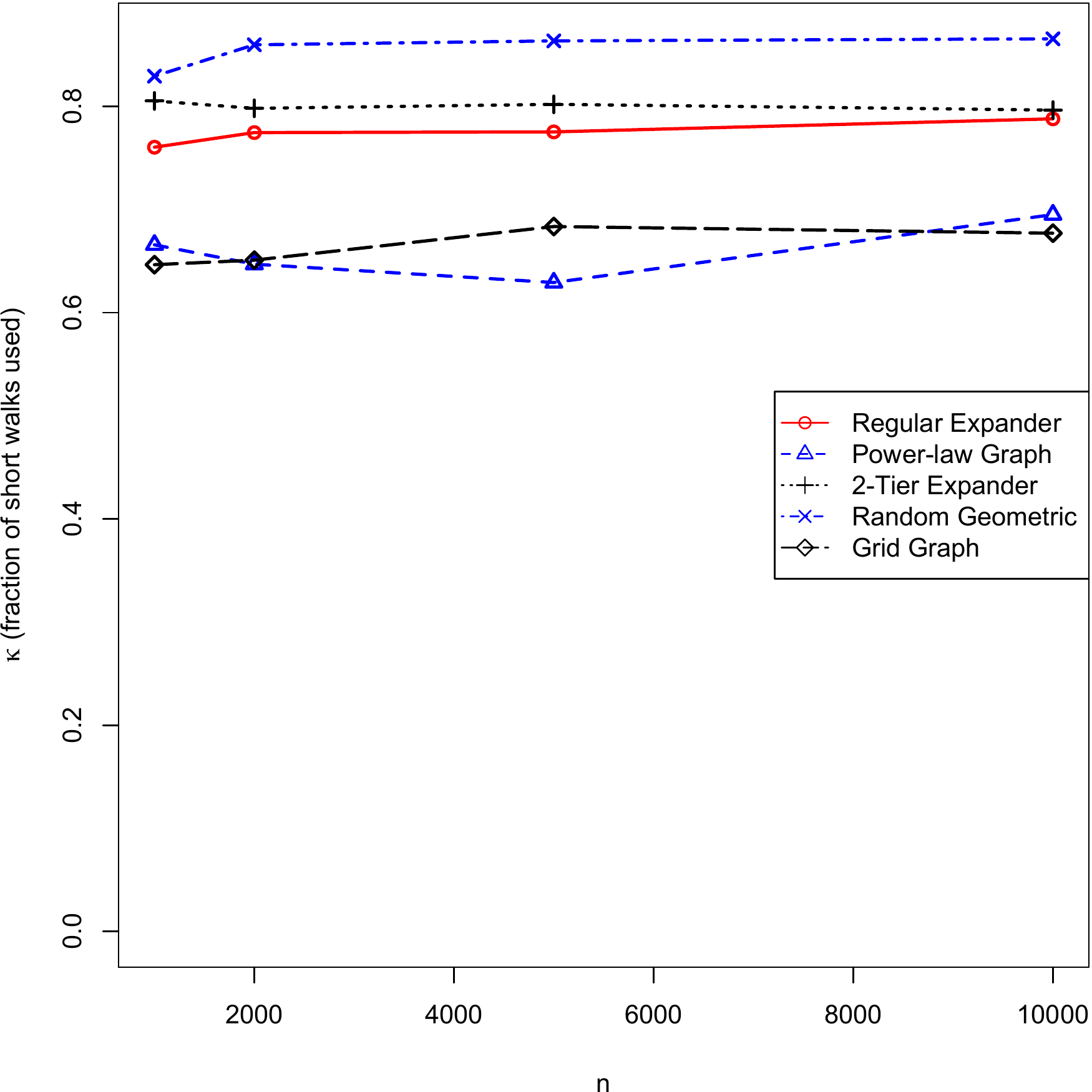}\\
  \caption{varying number of nodes n. $\ell =n$, $\eta = 1, \lambda = \sqrt{\ell}$}
  \label{fig:plot2}
\end{figure}

\noindent{\bf Varying $n$} [Figure~\ref{fig:plot2}]: The number of nodes $n$ is varying between 1000 and 10,000. We see that in all of the graphs, the utilization of pre-processed short walks before the algorithm terminates is at least 60\%. We also see that in some of the graphs, the utilization is substantially higher. There even as the graph size scales, our performance remains equally good. 

\begin{figure}[htbp]
\centering
\includegraphics[width = 0.6\linewidth]{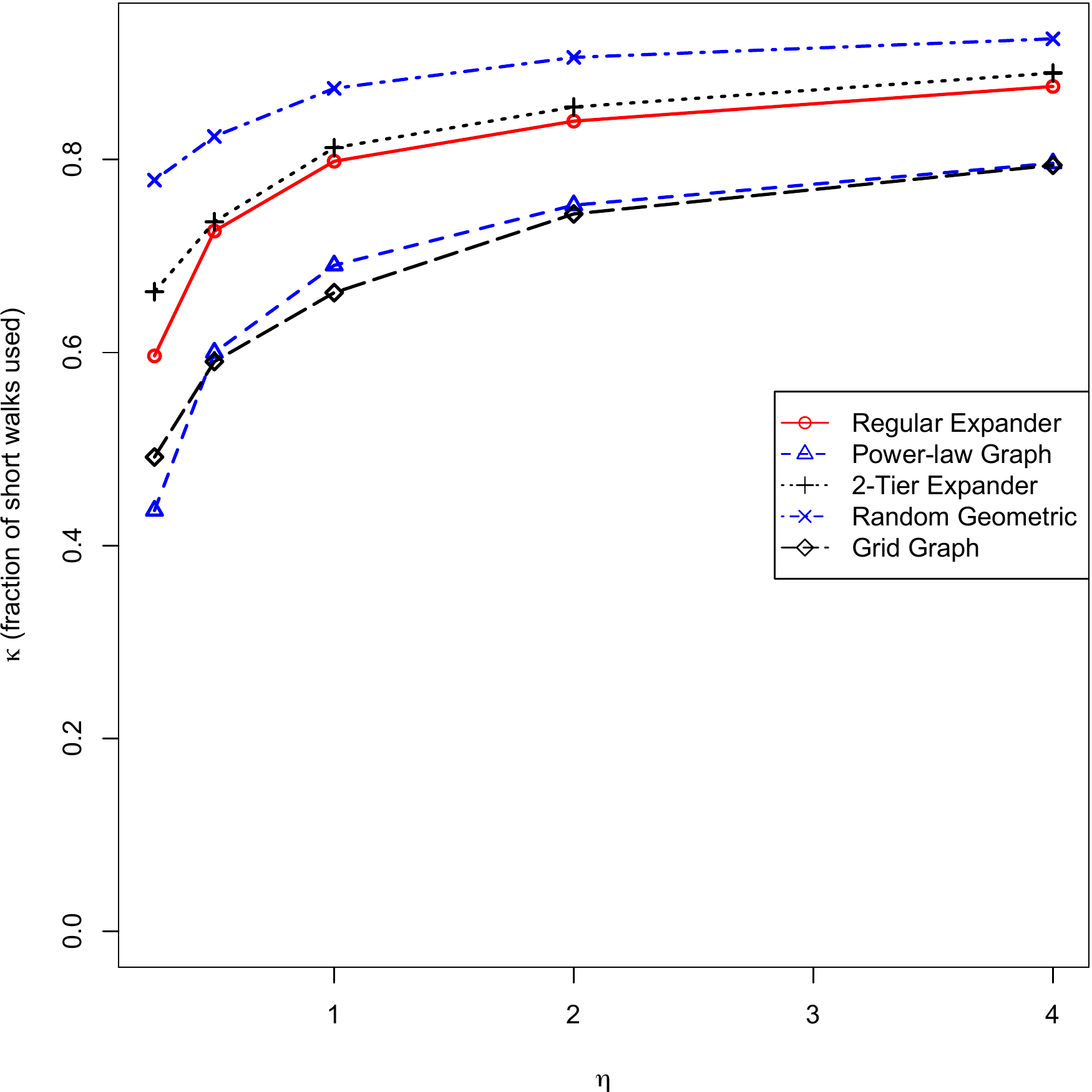}\\
\caption{varying number of short walks $\eta$. $n = 10K, \ell =n, \lambda = \sqrt{\ell}$}
\label{fig:plot3}
\end{figure}

\noindent{\bf Varying $\eta$} [Figure~\ref{fig:plot3}]: We see that the used fraction of rows is increasing with the number of short length walk $\eta$. We see that even for small enough $\eta$ of 1, on all algorithms, the utilization $\kappa$ on the $y$-axis is at least $0.6$, or 60\% of all the short walks get used. This means that for each node $v$, the number of short length walk $d(v)\log n$ suffice, therefore the round and message complexity remain near-optimal as proved previously.  

\begin{figure}[htbp]
  \centering
  \includegraphics[width=0.6\linewidth]{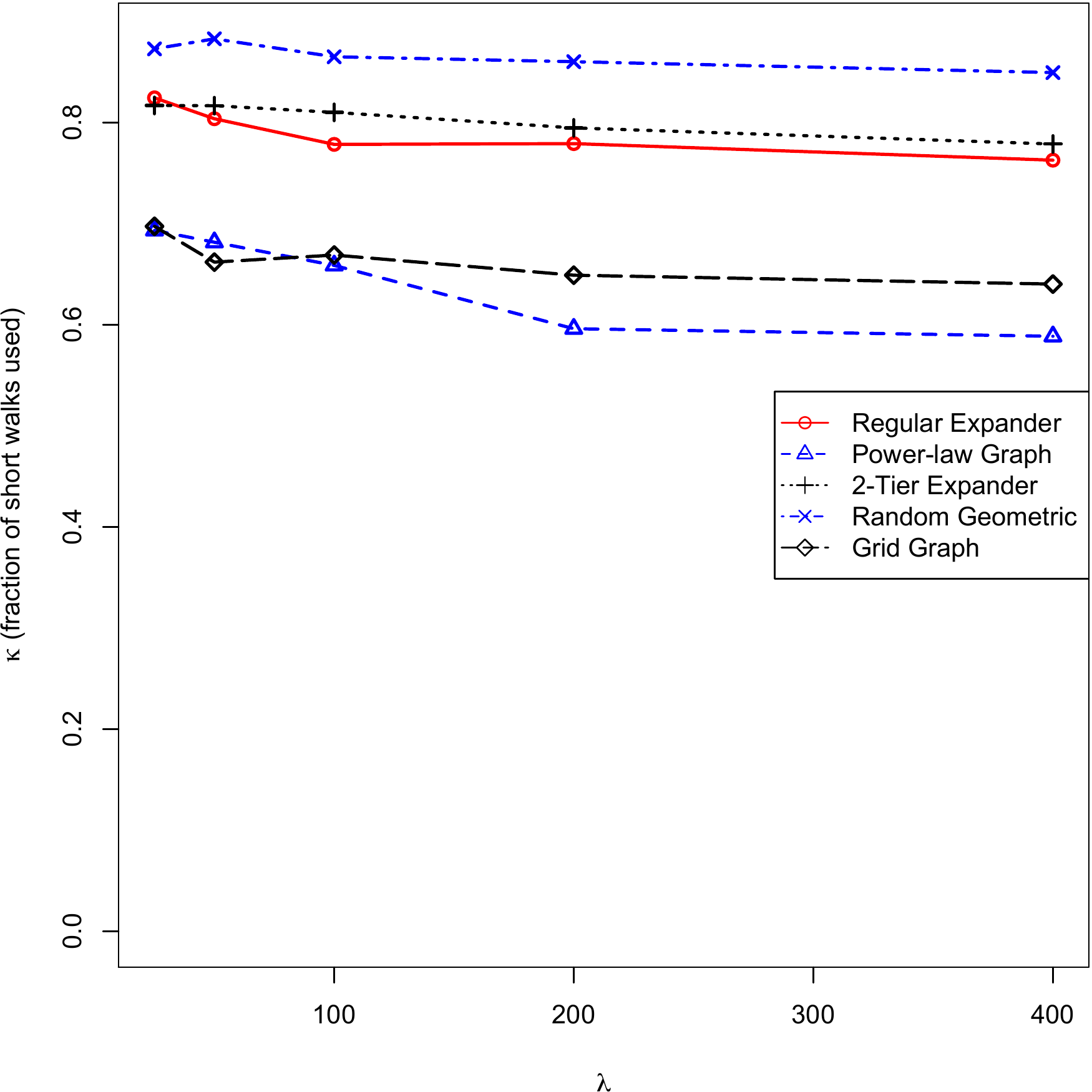}\\
  \caption{varying length of short walk $\lambda$. $n = 10K, \ell =n, \eta = \log n$}
  \label{fig:plot4}
\end{figure}

\noindent{\bf Varying $\lambda$} [Figure~\ref{fig:plot4}]: The default value of $\lambda$ is $\sqrt{\ell}$. In this plot, we vary $\lambda$ from $0.25\sqrt{\ell}$ to $\sqrt{\ell}$ in doubling steps. We see that the utilization roughly remains the same throughout the plot. Even though the algorithm needs to choose $\lambda$ to optimize for rounds and messages, this plot shows that for any of the values, it performs well.

\noindent {\bf Summary of observed round complexity:} To summarize the plots for varying different parameters on the $x$-axis, we see that in all the plots, the value of $\kappa$ on the $y$-axis is a constant and usually at least $0.5$. Since $\kappa$ is $1$ for optimal or perfect utilization of the table, we see that for all parameter values, the utilization is only a small constant factor (around $2$) away from the optimal. Therefore, the round complexity, as proven in the previous section, increases only marginally. 
\newpage
\subsection{Message complexity plots}

\begin{figure}[htbp]
  \centering
  \includegraphics[width=0.6\linewidth]{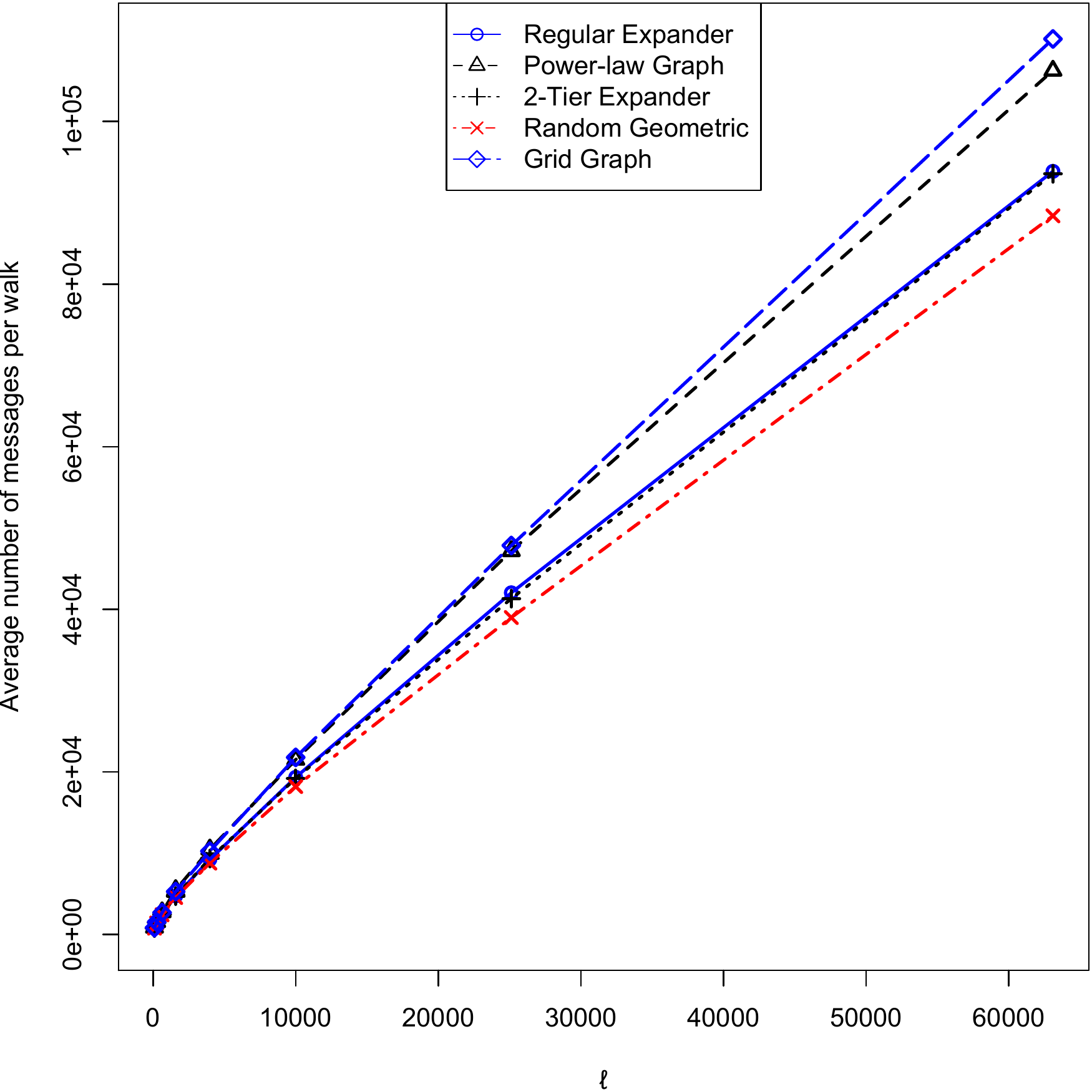}\\
  \caption{varying length of the walk $\ell$. $n=10K, \eta = 1, \lambda = \sqrt{\ell}$}
  \label{fig:Mplot1}
\end{figure}

\noindent{\bf Varying $\ell$} [Figure~\ref{fig:Mplot1}]: In this plot, we vary $\ell$ and note the message complexity of the algorithm {\sc Continuous-Random-Walk}, per random walk {\sc Single-Random-Walk} request within it. For any walk $\ell$ the optimal number of messages would be $\ell$ itself. Notice that in our plot also, all the lines (that is for all the graphs) are very close to the $x = y$ line, which is the optimal line. Therefore, the efficiency of {\sc Continuous-Random-Walk} amortized is almost the best possible.

\begin{figure}[htbp]
  \centering
  \includegraphics[width=0.6\linewidth]{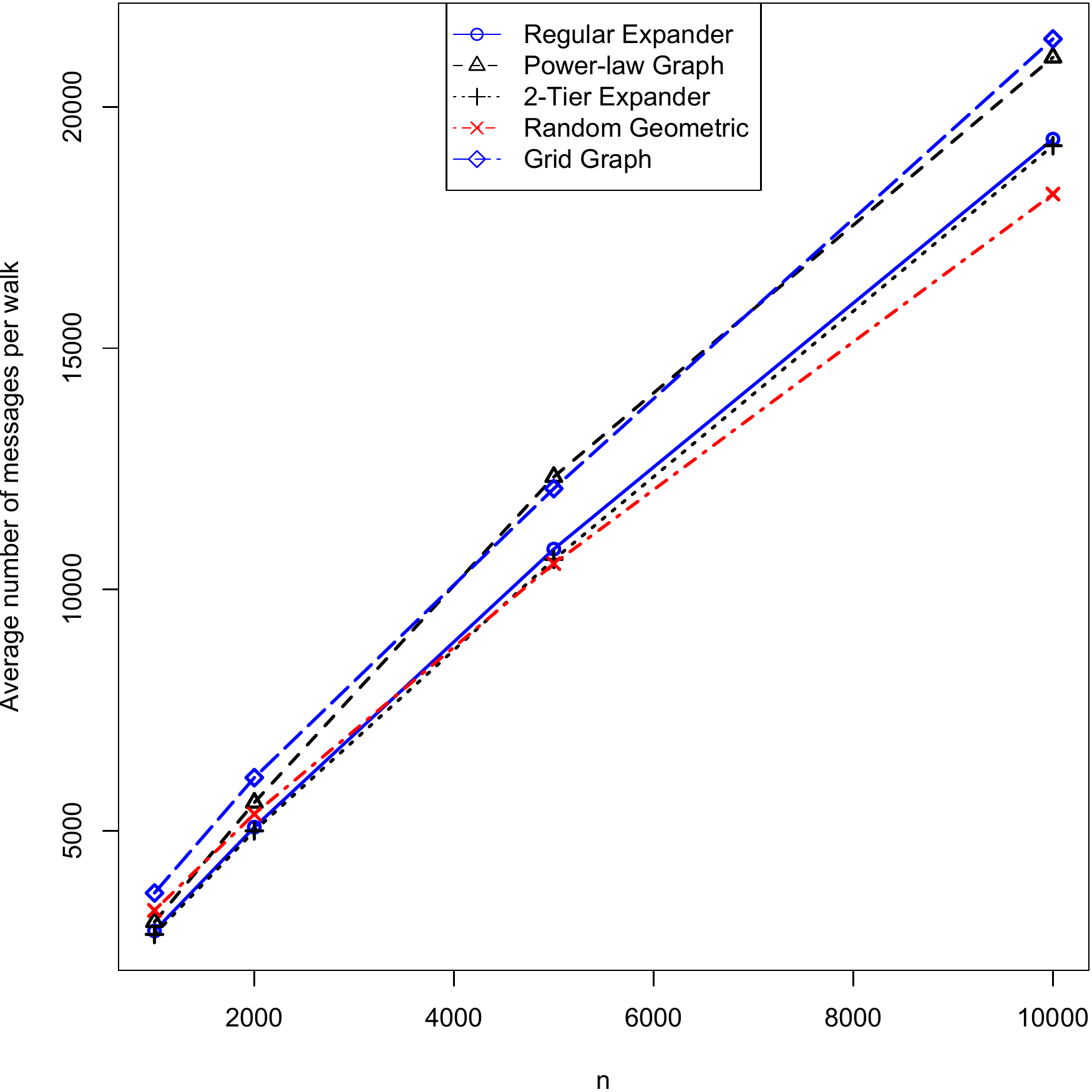}\\
  \caption{varying number of nodes n. $\ell =$n, $\eta = \mbox{ logn },\lambda = \sqrt{\ell}$}
  \label{fig:Mplot2}
\end{figure}

\noindent{\bf Varying $n$} [Figure~\ref{fig:Mplot2}]: In this plot as well, since we use the default value of $\ell = n$, the best possibility is for the message complexity to be $n$, which corresponds to the $x = y$ line. Notice that again for all the graphs, the lines for message complexity, through the entire range, is almost the best possible; this is because we get straight lines with the slope being very close to $x = y$ line. 

\begin{figure}[htbp]
\centering
\includegraphics[width = 0.6\linewidth]{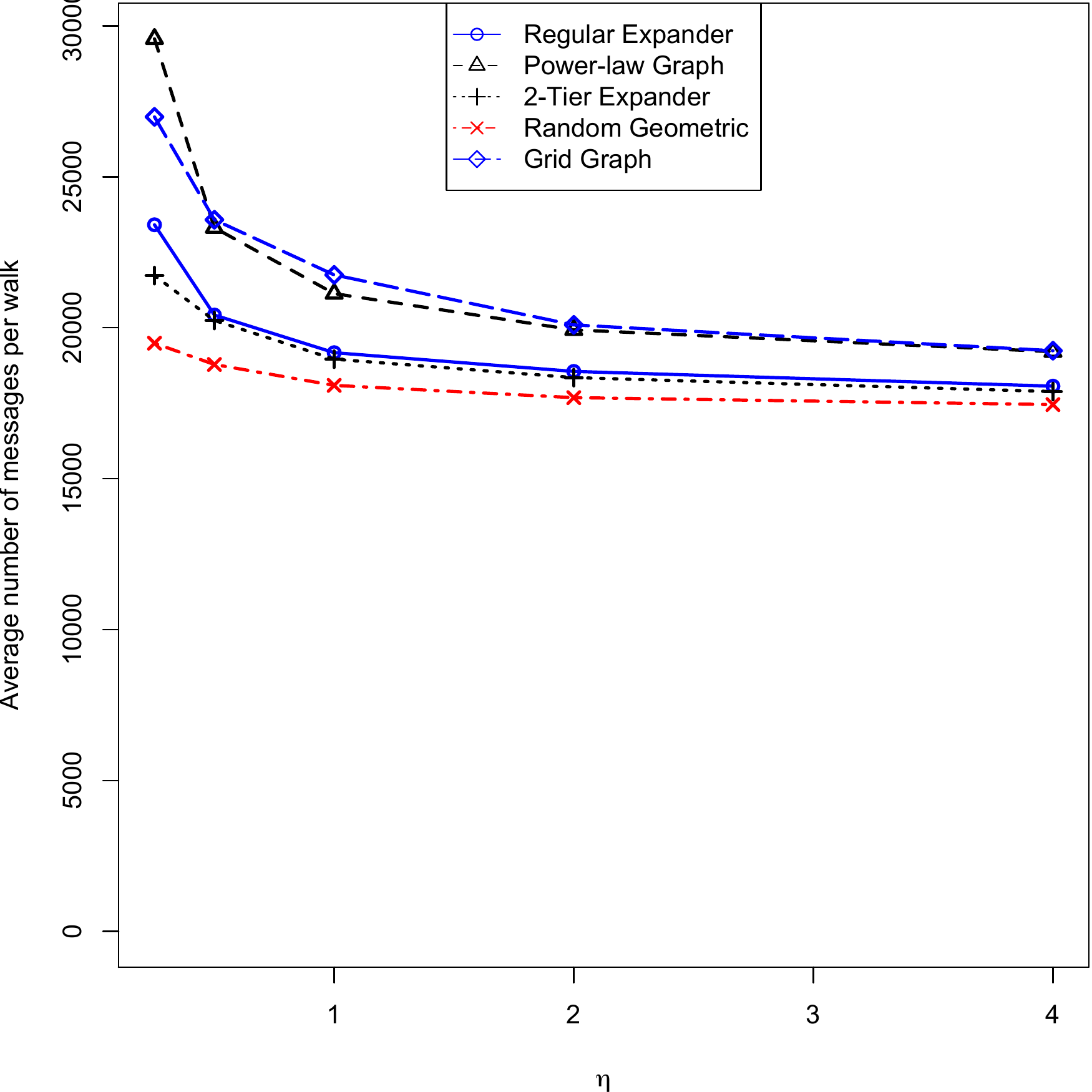}
\caption{varying number of short walks $\eta$. $n = 10K, \ell =n, \lambda = \sqrt{\ell}$}
\label{fig:Mplot3}
\end{figure}

\noindent{\bf Varying $\eta$} [Figure~\ref{fig:Mplot3}]: As $\eta$ is increased between $0.25$ and $4$, we see that the message complexity reduces rapidly. It is expected that as the number of pre-processing rows are increased, the efficiency would improve and therefore message complexity also improves. This plots sharp decline, however, also suggests that just a small enough $\eta$ is also sufficient to drastically bring down the message complexity close to optimal, regardless of what the graph topology is. 

\begin{figure}[htbp]
  \centering
  \includegraphics[width=0.6\linewidth]{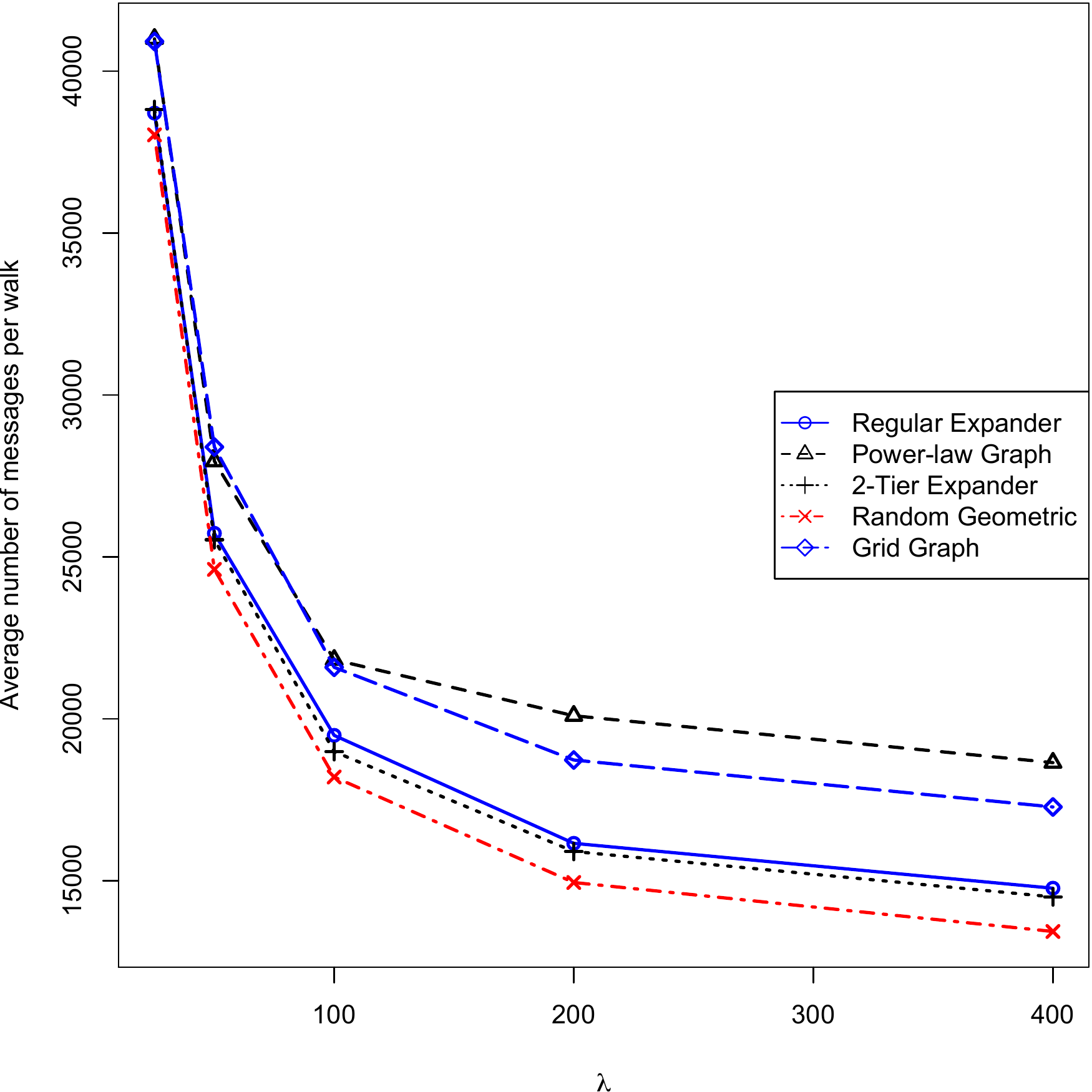}\\
  \caption{varying length of short walk $\lambda$. $n = 10K, \ell =n, \eta = \log n$}
  \label{fig:Mplot4}
\end{figure}

\noindent{\bf Varying $\lambda$} [Figure~\ref{fig:Mplot4}]: This plot is very similar to that of varying $\eta$, here we see again that as $\lambda$ is increased, the message complexity goes down rapidly. Recall that here we are comparing different $\lambda$ values for a fixed $\ell$ value of $n$. Our algorithm {\sc Continuous-Random-Walk} uses $\lambda = \sqrt{\ell}$ but we tried this plot with even smaller values of $\lambda$. As expected, the message complexity is high initially, however, as $\lambda$ is increased close to $\sqrt{\ell}$, the message complexity rapidly reduces, and improves the algorithm performance substantially.

\noindent {\bf Summary of observed message complexity:} In the naive approach, while each random walk requires $O(\ell)$ messages, the round complexity is increased significantly. At the other extreme, each random walk in~\cite{DasSarmaNPT10} was round-efficient but required $\Omega(m)$ messages! Our algorithm of {\sc Continuous-Random-Walk} achieves the best of both worlds by guaranteeing best-possible message and round complexity for graph topologies. The experiments suggest that for a wide range of parameters, the algorithm is able to answer each {\sc Single-Random-Walk} request in a continuous manner with very few messages or rounds. These results corroborate our theoretical guarantees and highlight the practicality of our technique.

\section{Conclusion}\label{sec:conclusion}
We present near-optimal distributed algorithms for random walk sampling in networks.
Since node sampling is useful in various networking applications,
our algorithms can serve as building blocks in a variety of distributed networking applications.

\newpage

  \let\oldthebibliography=\thebibliography
  \let\endoldthebibliography=\endthebibliography
  \renewenvironment{thebibliography}[1]{%
    \begin{oldthebibliography}{#1}%
      \setlength{\parskip}{0ex}%
      \setlength{\itemsep}{0ex}%
  }%
  {%
    \end{oldthebibliography}%
  }
{\small
\bibliographystyle{abbrv}
\bibliography{Distributed-RW}
}

\end{document}